\documentclass[12pt]{article}
\usepackage{amsmath}
\usepackage{amssymb}
\usepackage{amsfonts}
\usepackage{latexsym}
\usepackage{color}
\usepackage{graphicx}

\catcode `\@=11 \@addtoreset{equation}{section}

\catcode `\@=12

\newcommand{\be}{\begin{equation}}
\newcommand{\en}{\end{equation}}
\newcommand{\bea}{\begin{eqnarray}}
\newcommand{\ena}{\end{eqnarray}}
\newcommand{\beano}{\begin{eqnarray*}}
\newcommand{\enano}{\end{eqnarray*}}
\newcommand{\bee}{\begin{enumerate}}
\newcommand{\ene}{\end{enumerate}}

\newcommand{\Lc}{{\cal L}}
\newcommand{\Sc}{{\cal S}}

\newcommand{\1}{1 \!\! 1}
\newtheorem{thm}{Theorem}

\newenvironment{proof}{\noindent {\bf Proof:}}{\hfill$\Box$}

\begin{document}
%\nocite{*}
\thispagestyle{empty}

\vspace*{1cm}

\begin{center}
{\Large \bf Fourier transforms, fractional derivatives, and a little bit of quantum mechanics}   \vspace{2cm}\\

{\large F. Bagarello}
%\footnote[1]{ Dipartimento di Matematica ed Applicazioni,
%Fac.Ingegneria, Universit\`a di Palermo, I - 90128  Palermo, Italy}
\vspace{3mm}\\[0pt]
Dipartimento di Ingegneria, Scuola Politecnica,\\[0pt]
Universit\`{a} di Palermo, I - 90128 Palermo,  and\\
\, INFN, Sezione di Napoli, Italy.\\[0pt]
E-mail: fabio.bagarello@unipa.it\\[0pt]
home page: www1.unipa.it/fabio.bagarello \vspace{8mm}\\[0pt]
\end{center}

\vspace*{0.5cm}

\begin{abstract}
We discuss some of the mathematical properties of the fractional derivative defined by means of Fourier transforms. We first consider its action  on the set of test functions $\Sc(\mathbb R)$, and then we extend it to its dual set, $\Sc'(\mathbb R)$, the set of tempered distributions, provided they satisfy some mild conditions. We discuss some examples, and we show how our definition can be used in a quantum mechanical context.
\end{abstract}

\vspace{2cm}


\vfill

\newpage

% Section 1
\section{Introduction}

Since many years a lot of people have been interested in defining a fractional version of the derivative, and work with it. The interest was both purely mathematical and for possible applications to physics. A very partial list of references includes the monographes or edited books \cite{hilfer}-\cite{laskinbook}, and the following papers, which are among the few which adopt fractional derivatives of different kind in quantum mechanics, \cite{namias}-\cite{wei2}. This is the only kind of applications we will consider here. We refer to the bibliographies in the books cited above for more references and for more applications to different realms of physics and engineering.

It is well known that it does not exist a unique definition of the fractional derivative. In fact, there are many of them,  introduced by different people, in different contexts, and having different peculiarities. Some well known definitions are due to J. Liouville, G. F. B. Riemann, N. Y. Sonin, and, more recently, to M. Caputo. It may be worth stressing that different definitions may give rise to different results when computing the derivatives of the same function, even if all possible extensions return the standard outputs  when "fractional"  is replaced by "usual": they all extend, in different ways, the ordinary derivative.

As we have already stressed, not many applications of fractional derivatives to quantum mechanics can be found in the literature, except those cited above:
 in \cite{laskin1} the author proposes his general view to what he calls {\em fractional quantum mechanics}, while in \cite{laskin2} the same author proposes a fractional version of the Schr\"odinger equation. In \cite{matos} the author discusses an application of fractional derivatives to quantum mechanics, analysing the possibility of defining a fractional momentum operator in one dimension, while in \cite{wei1,wei2}, among other topics, the author discusses the Heisenberg uncertainty relation in this extended settings. A more systematic use of fractional derivatives in quantum mechanics is discussed in \cite{herr}.

What is particularly interesting for us is the possibility of extending the canonical commutation relations to the case in which the ladder operators involved extend the standard ones. In \cite{baginbagbook} this idea has been carried out considering ladder operators which, contrarily to what happens for ordinary bosonic operators, are not related by the adjoint operation. Here we want to consider a slightly different extension, replacing the space derivative appearing in the coordinate representation of the creation and annihilation operators\footnote{These are the operators $\frac{1}{\sqrt{2}}\left(x-\frac{d}{dx}\right)$ and $\frac{1}{\sqrt{2}}\left(x+\frac{d}{dx}\right)$, respectively. } with a fractional derivative. This will be done in Section \ref{sect5}, where other aspects of the fractional momentum operators will be also discussed, including the related uncertainty relation. The definition of fractional derivative is given in Section \ref{sect2}, together with some properties. The definition we use essentially coincides with that considered in \cite{herr,tseng} and \cite{vigue}, for instance. With respect to  \cite{herr,tseng}, we will be more concerned with the mathematical aspects of the fractional derivative. Also, rather than adopting a {\em complex analysis} approach as in \cite{vigue}, we will use ideas coming from functional analysis and distribution theory. This is the content of Section \ref{sect2}. Examples and applications of our general results are given in Section \ref{sect3}, while, as stated,  Section \ref{sect5} contains the physical consequences of our definition. In Section \ref{sect6} we give our conclusions and our plans for the future.

\section{Definition and first properties}\label{sect2}

The idea of our approach is very simple and it is directly connected with a well known property of the Fourier transform. Let $f(x)\in\Sc(\mathbb R)$. Its Fourier transform $\hat f(p)=\frac{1}{\sqrt{2\pi}}\int_{\mathbb R}\,e^{-ipx}f(x)dx$ belongs to $\Sc(\mathbb R)$ as well, and satisfies the following equality:
$$
F\left[\frac{d^nf(x)}{dx^n}\right](p)=(ip)^n\hat f(p),
$$
for all $n=0,1,2,3,\ldots$. Here we are adopting the (standard) notation: $F[g](p)=\hat g(p)$, for all $g(x)$ which admits Fourier transform. Hence $F^{-1}$ indicates the inverse Fourier transform, so that
\be
\frac{d^nf(x)}{dx^n}=F^{-1}\left[(ip)^n\hat f(p)\right]=\frac{1}{\sqrt{2\pi}}\int_{\mathbb R}\,e^{ipx}(ip)^n\hat f(p)dp,
\label{21}\en
and the integral is surely well defined.
This formula can be generalized to define the fractional derivative $D^{\alpha}$, $\alpha\in\mathbb{R}_+$ as follows:
\be
D^\alpha f(x)=F^{-1}\left[(ip)^\alpha\hat f(p)\right]=\frac{1}{\sqrt{2\pi}}\int_{\mathbb R}\,e^{ipx}(ip)^\alpha\hat f(p)dp,
\label{22}\en
whenever this integral exists\footnote{Here and in the following, to fix uniquely the value of $i^\alpha$, we take $i^\alpha=\cos\left(\frac{\alpha\pi}{2}\right)+i\sin\left(\frac{\alpha\pi}{2}\right)$.}, see \cite{herr,tseng,vigue}.  This is what the next theorem is about. 

\begin{thm}
	If $f(x)\in\Sc(\mathbb{R})$ then $D^\alpha f(x)$ is well defined and belongs to $\Lc^\infty(\mathbb R)\cap\Lc^2(\mathbb R)$ for all $\alpha\geq0$.
\end{thm}

\begin{proof}
	Since $f(x)\in\Sc(\mathbb{R})$, its Fourier transform  $\hat f(p)$ is also in $\Sc(\mathbb{R})$. Let $g_\alpha(p)=(ip)^\alpha \hat f(p)$. Since $|g_\alpha(p)|^2=|p|^{2\alpha} |\hat f(p)|^2$, it is clear that $g_\alpha(p)\in\Lc^1(\mathbb R)\cap \Lc^2(\mathbb R)$. Then, according to the definition, $D^\alpha f(x)$ is the inverse Fourier transform of $g_\alpha(p)$. This implies that $D^\alpha f(x)\in\Lc^\infty(\mathbb R)\cap\Lc^2(\mathbb R)$, as we had to prove.
	
\end{proof}

\vspace{2mm}

{\bf Remarks:--} (1) $D^0$ is the identity operator on $\Sc(\mathbb{R})$. In fact, from (\ref{22}), $D^0f(x)=F^{-1}\left[\hat f(p)\right]=f(x)$, for all $f(x)\in\Sc(\mathbb{R})$. Moreover, Plancherel's theorem states that the equality $D^0=\1$ can be extended to $\Lc^2(\mathbb{R})$.

(2) Similar arguments show  that $g_\alpha(p)\in\Lc^r(\mathbb R)$ for all $r\geq1$; 

(3) It is also clear that, if $\alpha$ takes integer values $n=0, 1, 2, 3$, and so on, $D^\alpha$ coincides with the standard derivative: $D^n=\frac{d^n}{dx^n}$.

\vspace{2mm}

The fractional derivative $D^\alpha$ on a product of test functions, $f(x)g(x)$, $f(x), g(x)\in\Sc(\mathbb R)$, can be computed using the fact that $F[fg](p)=\frac{1}{\sqrt{2\pi}}(\hat f\ast \hat g)(p)$. Simple computations show that
\be
D^\alpha(fg)(x)=\frac{i^\alpha}{2\pi}\int_{\mathbb R} e^{isx}\hat g(s)\left(\int_{\mathbb R} e^{iqx}\hat f(q)(s+q)^\alpha\,dq\right)ds,\label{23}\en
which, in particular, returns $D^n(fg)(x)=\frac{d^n}{dx^n}(fg)(x)$, for $n=0,1,2,\ldots$. Hence, formula (\ref{23}) extends the standard Leibnitz rule for derivatives of products. For non integer $\alpha$, the computation is much harder and there is no particularly easy way to compute the result, since the binomial expansion of $(s+q)^\alpha$  depends on the ratio $\frac{s}{q}$, which of course changes when computing the integrals.

In view of our applications in Section \ref{sect5}, it is interesting to prove the following equality:
\be
\left<D^\alpha f,g\right>=(-1)^\alpha \left< f, D^\alpha g\right>,
\label{23b}\en
for all $f(x)$ and $g(x)$ in $\Sc(\mathbb R)$, and for all $\alpha\geq0$. In fact we have
$$
\left<D^\alpha f,g\right>=\left<F^{-1}\left[(ip)^\alpha\hat f(p)\right],F^{-1}\left[\hat g(p)\right]\right>=\left<(ip)^\alpha\hat f(p),\hat g(p)\right>,
$$
where we have used the well known Parceval equality for $\Lc^2$-functions. Analogously we find
$$
\left< f,D^\alpha g\right>=\left<F^{-1}\left[\hat f(p)\right],F^{-1}\left[(ip)^\alpha\hat g(p)\right]\right>=\left<\hat f(p),(ip)^\alpha\hat g(p)\right>,
$$
from which (\ref{23b}) follows.

\vspace{2mm}

The operator $D^\alpha$, defined on $\Sc(\mathbb{R})$, can be extended to a larger set of functions, at least in a weak sense. The idea comes from distribution theory. It is known that the ordinary derivative $\frac{d}{dx}$, and its powers, can be extended to the set $\Sc'(\mathbb R)$  of tempered distributions by duality: for all $f(x)\in\Sc(\mathbb{R})$ and $F(x)\in\Sc'(\mathbb{R})$, one can define $\frac{d^n}{dx^n}F(x)$ using
\be\left<\frac{d^n}{dx^n}F,f\right>=(-1)^n\left<F,\frac{d^n}{dx^n}f\right>,
\label{24}\en
for all $n=0,1,2,3,\ldots$. Here $\left<.,.\right>$ is the form which puts in duality $\Sc(\mathbb R)$
and $\Sc'(\mathbb R)$, and it extends the ordinary scalar product in $\Lc^2(\mathbb R)$. It is well known that, with this definition, $\frac{d^n}{dx^n}F\in\Sc'(\mathbb R)$: the weak derivatives of a tempered distributions are also tempered distributions. Following this strategy, we define
\be\left<D^\alpha\Psi,f\right>:=(-1)^\alpha\left<\Psi,D^\alpha f\right>,
\label{25}\en
for all $\alpha\geq0$, $\Psi(x)\in\Lc^2(\mathbb R)$ and $f(x)\in\Sc(\mathbb R)$. This is well defined, since $\Psi(x), D^\alpha f(x)\in\Lc^2(\mathbb{R})$. Notice that this definition extends formula (\ref{23b}) to pairs of functions which are not both in $\Sc(\mathbb R)$.
 We refer to \cite{samko} for a rather rich analysis of the connections between fractional derivatives and distribution theory. This definitions produce the following interesting continuity result:

\begin{thm}
	If $\{f_n(x)\in\Sc(\mathbb{R})\}$ is a sequence of test functions such that $\tau_\Sc-\lim_{n,\infty}f_n(x)=f(x)$, $f(x)\in\Sc(\mathbb{R})$, then $\left<D^\alpha\Psi,f_n\right>\rightarrow \left<D^\alpha\Psi,f\right>$, for all $\Psi(x)\in\Lc^2(\mathbb R)$.
\end{thm}

\begin{proof}
	First we recall that $\tau_\Sc-\lim_{n,\infty}f_n(x)=f(x)$ means that $\lim_{n,\infty}\sup_{x\in\mathbb{R}}|x^l(f_n^{(k)}(x)-f^{(k)}(x)|=0$, for all $k,l=0,1,2,3,\ldots$. Now, to prove the statement, it is enough to check that $\|D^\alpha f_n-D^\alpha f\|\rightarrow0$ for $n\rightarrow\infty$, where $\|.\|$ is the norm in $\Lc^2(\mathbb{R})$. In fact, in this case, we have
	$$
	\left|\left<D^\alpha\Psi,f_n\right>-\left<D^\alpha\Psi,f\right>\right|=\left|\left<\Psi,D^\alpha f_n\right>-\left<\Psi,D^\alpha f\right>\right|\leq \|\Psi\|\|D^\alpha f_n-D^\alpha f\|\rightarrow0
	$$
	which is what we have to prove. Using formula (\ref{22}), simple computations show that
	$$
	\|D^\alpha f_n-D^\alpha f\|^2=\int_{\mathbb R}
|p|^{2\alpha}\left|\hat f_n(p)-\hat f(p)\right|^2\,dp.	$$
Now, since $f_n(x)$ is $\tau_\Sc$-convergent to $f(x)$, $\hat f_n(p)=F[f_n](p)$ converges to $\hat f(p)=F[f](p)$ in the same topology. Hence,
$$
\lim_{n,\infty}|p|^l\left(\hat f_n^{(k)}(p)-\hat f^{(k)}(p)\right)=0,
$$
uniformly in $p$, for all $l,k=0,1,2,3,\ldots$. Now, let us write $\|D^\alpha f_n-D^\alpha f\|^2$ as the following sum of integrals:
 $$
 \|D^\alpha f_n-D^\alpha f\|^2=\int_{|p|\leq1}
 |p|^{2\alpha}\left|\hat f_n(p)-\hat f(p)\right|^2\,dp+$$
 $$+\int_{|p|>1}
 |p|^{2\alpha}\left|\hat f_n(p)-\hat f(p)\right|^2\,dp=I_n^{(|p|\leq1)}+I_n^{(|p|>1)},	$$
with obvious notation. Since $I_n^{(|p|\leq1)}\leq \int_{|p|\leq1}
\left|\hat f_n(p)-\hat f(p)\right|^2\,dp,$ and since, in particular, $\hat f_n(p)$ converges uniformly to $\hat f(p)$, it is clear that $I_n^{(|p|\leq1)}\rightarrow0$ for $n\rightarrow\infty$. To prove that the same holds true for $I_n^{(|p|>1)}$, we first observe that, for all $l\geq \alpha$, $l\in\mathbb{N}$, 
$$
|p|^{2\alpha}\left|\hat f_n(p)-\hat f(p)\right|^2\leq |p|^{2l}\left|\hat f_n(p)-\hat f(p)\right|^2,
$$
since $|p|>1$. The right-hand side of this inequality converges to zero, uniformly in $p$, for each integer value of $l$. Then, in particular, $\forall\epsilon>0$ there exists $n_\epsilon(l+1)\in\mathbb{N}$ such that $|p|^{2(l+1)}\left|\hat f_n(p)-\hat f(p)\right|^2\leq\epsilon$, for all $n\geq n_\epsilon(l+1)$ and for all $p\in\mathbb{R}$. We use $n_\epsilon(l+1)$ to emphasize the fact that the index depends not only on $\epsilon$, but also on the particular power of $|p|$. However, this dependence is not particularly relevant for us, since $l$ is a fixed integer in our estimate. With this in mind we see that, for all $p\in\mathbb{R}$,
$$
|p|^{2\alpha}\left|\hat f_n(p)-\hat f(p)\right|^2\leq |p|^{2l}\left|\hat f_n(p)-\hat f(p)\right|^2\leq \frac{\epsilon}{|p|^2},
$$
and therefore
$$
I_n^{(|p|>1)}\leq \int_{|p|>1}|p|^{2l}\left|\hat f_n(p)-\hat f(p)\right|^2dp\leq 2\epsilon\int_1^\infty \frac{dp}{p^2}=2\epsilon,
$$
which must be true for $n$ sufficiently large. How large $n$ should be depends, of course, on $l$. But for each fixed $l$, we can conclude that $\lim_{n,\infty} I_n^{(|p|>1)}=0$. This, together with the analogous conclusion on $I_n^{(|p|\leq1)}$, allows us to conclude the proof.

\end{proof}

\vspace{2mm}

It is possible to extend further the definition of $D^\alpha$ to other functions (or even to distributions) which are not necessarily in $\Lc^2(\mathbb R)$. Let $G(x)$ be such that its Fourier transform exists (also in a distributional sense, if needed), and suppose $G(x)$ satisfies the following existence requirement:
\be
I_\alpha=\int_{\mathbb R}\overline{\hat G(p)}\,(ip)^\alpha \hat f(p)\,dp
\label{25bis}\en
exists for all $f(x)\in\Sc(\mathbb{R})$ and (at least) for some $\alpha\geq0$. Then we can extend (\ref{25}) as follows:
\be\left<D^\alpha G,f\right>:=(-1)^\alpha\left<G,D^\alpha f\right>,
\label{26}\en
for those $\alpha$ and $G(x)$, and for all $f(x)\in\Sc(\mathbb R)$. In fact, using (\ref{22}) and the Parceval equality (extended to these functions), we have
$$
\left<G,D^\alpha f\right>=\left<F[G],F[D^\alpha f]\right>=\left<\hat G,(ip)^\alpha \hat f\right>=I_\alpha,
$$ 
which exists by assumption. Hence, (\ref{26}) makes sense, since the right-hand side is well defined.

 \vspace{2mm}
 
\noindent{\bf Remark:--} It is easy to find conditions for $I_\alpha$ to be well defined. It is sufficient that $\overline{\hat G(p)}\,(ip)^\alpha$ is a tempered distribution, i.e., an element of $\Sc'(\mathbb{R})$. This is granted, for instance, if $\alpha=0,1,2,3,\ldots$ and if $G(x)\in\Sc'(\mathbb{R})$. We will go back to the existence of $I_\alpha$ in the next section, for specific examples.

\vspace{2mm}

Another interesting aspect of the operator $D^\alpha$ is given by the following theorem, stating essentially that, if $\alpha$ is {\em close} to an integer $n$, then $D^{\alpha}f(x)$ is {\em close} to $D^nf(x)$, for suitable functions. 

\begin{thm}\label{th3}
	Given any $\{f(x)\in\Sc(\mathbb{R})\}$, for all fixed $n=0,1,2,3,\ldots$, the sequence $D^{n+\frac{1}{k}}f(x)$ converges  to $D^{n}f(x)$, when $k\rightarrow\infty$ uniformly in $x$.
\end{thm}

\begin{proof}
	First we observe that
	$$
	D^{n+\frac{1}{k}}f(x)-D^{n}f(x)=\frac{1}{\sqrt{2\pi}}\int_{\mathbb R}h_k^{[n,x]}(p)\,dp,
	$$
	where 
	$$
	h_k^{[n,x]}(p)=e^{ipx}(ip)^n\left((ip)^{\frac{1}{k}}-1\right)\hat f(p).
	$$
It is clear that, for all fixed $p$, $h_k^{[n,x]}(p)\rightarrow0$ when $k$ diverges. This is true for all values of $x\in\mathbb{R}$ and for all fixed $n\in\mathbb{N}$. It is further easy to find a positive function $\Phi^{[n]}(p)\in\Lc^1(\mathbb{R})$ such that $|h_k^{[n,x]}(p)|\leq \Phi^{[n]}(p)$ almost everywhere in $p$. Hence, from the Lebesgue dominated convergence theorem our claim follows.

To find the function $\Phi^{[n]}(p)$ we first observe that $$\left|(ip)^{\frac{1}{k}}-1\right|\leq |p|^{\frac{1}{k}}+1\leq \left\{
\begin{array}{ll}
2\qquad\qquad\,\, |p|\leq1, \\
2|p|\qquad\quad \mbox{elsewhere}
\end{array}
\right.$$
If we call $\rho(p)$ the function in the right-hand side of this inequality, we find that
$$
|h_k^{[n,x]}(p)|\leq |p|^n\rho(p)|\hat f(p)|=:\Phi^{[n]}(p),
$$
which is clearly integrable in $\mathbb{R}$.

\end{proof}

We conclude this section stressing that, even if the definition in (\ref{22}) was already known in the literature, not many existence and continuity results can still be found, as those given here. For this reason we believe  that this paper add something to the existence literature on the topic.

\section{Examples}\label{sect3}

This section is devoted to the analysis of some examples of the definition in (\ref{22}). We will also discuss some cases in which the function we want to derive is not in $\Sc(\mathbb R)$, so that we need to use formulas (\ref{25}) or (\ref{26}). It may be useful to stress that, when compared with \cite{tseng}, our approach is much more {\em mathematically oriented}, since we are not interested in getting purely formal results. 

\subsection{Example  1: $f_1(x)=e^{-x^2}$}

We first observe that $f_1(x)\in\Sc(\mathbb R)$. Its Fourier transform is again a Gaussian: $\hat f_1(p)=\frac{e^{-p^2/4}}{\sqrt{2}}$. Hence, $g_\alpha(p)=(ip)^\alpha \hat f_1(p)=(ip)^\alpha\,\frac{ e^{-p^2/4}}{\sqrt{2}}$, and (\ref{22}) becomes
$
D^\alpha f(x)=\frac{1}{2\,\sqrt{\pi}}\,J_\alpha,
$
where
$$
J_\alpha=\int_{\mathbb{R}}\,e^{ipx-p^2/4}(ip)^\alpha dp,
$$
which can be computed in terms of Gamma and Confluent Hypergeometric functions. We get
\be
D^\alpha f_1(x)=\frac{2^{\,\alpha}}{\sqrt{\pi}}\,\biggl[\cos\left(\frac{\alpha\pi}{2}\right)\Gamma\left(\frac{1+\alpha}{2}\right) {}_1F_1\left(\frac{1+\alpha}{2},\frac{1}{2},-x^2\right)-
$$
$$
-\alpha x\sin\left(\frac{\alpha\pi}{2}\right)\Gamma\left(\frac{\alpha}{2}\right) {}_1F_1\left(1+\frac{\alpha}{2},\frac{3}{2},-x^2\right)\biggr].
\label{31}\en
This result reflects a similar one in Table 5.1, \cite{herr}.
Now, since $\Gamma\left(\frac{1}{2}\right)=\sqrt{\pi}$, $\Gamma\left(\frac{3}{2}\right)=\frac{\sqrt{\pi}}{2}$, ${}_1F_1\left(\frac{1}{2},\frac{1}{2},-x^2\right)={}_1F_1\left(\frac{3}{2},\frac{3}{2},-x^2\right)=e^{-x^2}$, ${}_1F_1\left(\frac{3}{2},\frac{1}{2},-x^2\right)=(1-2x^2)e^{-x^2}$ and ${}_1F_1\left(\frac{5}{2},\frac{3}{2},-x^2\right)=\left(1-\frac{2x^2}{3}\right)e^{-x^2}$, it is simple to check explicitly that $D^nf_1(x)=f_1^{(n)}(x)$, at least for $n=0,1,2,3$, where $f_1^{(k)}(x)$ is the $k-$th ordinary derivative of $f_1(x)$. Of course, we expect no surprise for $\alpha$ integer and larger than 3. Figure \ref{fig1} shows the behaviour of $D^{\alpha}f_1(x)$ for different values of $\alpha$ between 0 and 1: the plots show clearly that, in complete agreement with Theorem \ref{th3}, the more $\alpha$ approaches zero, the closer the result is to $f_1(x)$. 

\begin{figure}[ht]
	\begin{center}
		\includegraphics[width=0.90\textwidth]{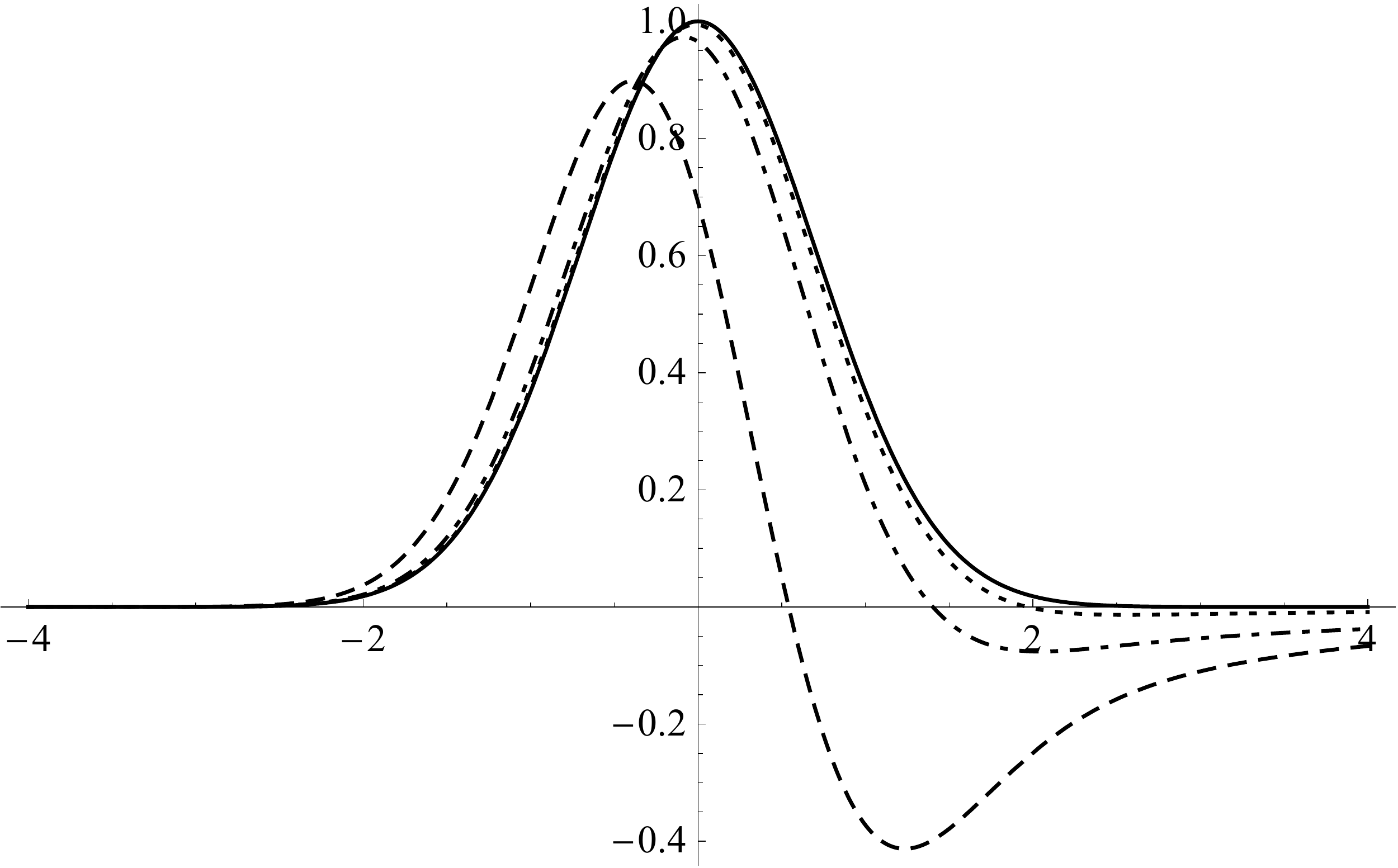}\hspace{8mm} %
	\end{center}
	\caption{{$D^\alpha f_1(x)$ for $\alpha=0$ (continuos line), $\alpha=\frac{1}{50}$ (dotted line), $\alpha=\frac{1}{10}$ (dotted-dashed line), and $\alpha=\frac{1}{2}$ (dashed line).}}
	\label{fig1}
\end{figure}

In Figure \ref{fig2} we compare what happens when $\alpha=5$ and when $\alpha$ approaches $5$ from below or from above. In particular we consider $\alpha=4.5, 4.9, 5, 5.1, 5.5$.

\begin{figure}[ht]
	\begin{center}
		\includegraphics[width=0.90\textwidth]{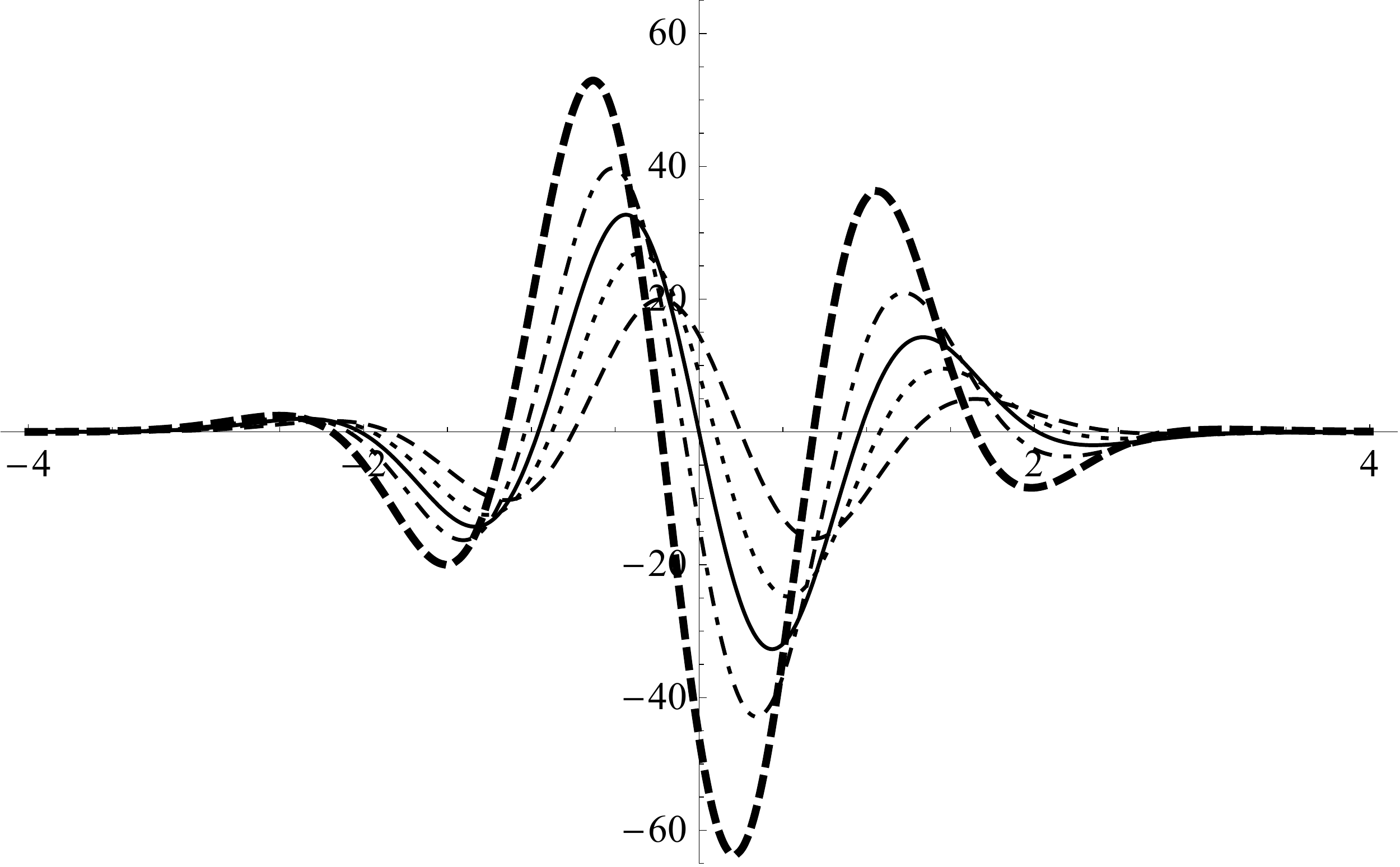}\hspace{8mm} %
	\end{center}
	\caption{{$D^\alpha f_1(x)$ for $\alpha=4.5$ (dashed line), $\alpha=4.8$ (dotted line), $\alpha=5$ (continuous line),  $\alpha=5.2$ (dotted-dashed line) and $5.5$ (thick-dashed line).}}
	\label{fig2}
\end{figure}
Both figures clearly show that, moving away from an integer value of $\alpha$, the fractional derivative in (\ref{22}) changes, and it approaches the {\em standard} result when the value of $\alpha$ is closer to the integer value. This is again an evidence of what discussed in Theorem \ref{th3}.  It is also interesting to observe that the original parity of the two derivatives, $D^0f_1(x)$ and $D^5f_1(x)$, is lost when $\alpha\neq0,5$: for instance, while $D^0f_1(x)$ is an even function, $D^\alpha f_1(x)$ is not, if $\alpha=\frac{1}{50}, \frac{1}{10}, \frac{1}{2}$.

\subsection{Example  2: $f_2(x)=x^2e^{-x^2}$}

This example is close to the previous one, in the sense that, also in this case, $f_2(x)\in\Sc(\mathbb R)$. Its Fourier transform is $\hat f_2(p)=\frac{(2-p^2)e^{-p^2/4}}{4\sqrt{2}}$, so that
$$
D^\alpha f_2(x)=\frac{1}{8\,\sqrt{\pi}}\,K_\alpha,
$$
where
$$
K_\alpha=\int_{\mathbb{R}}\,e^{ipx-p^2/4}(ip)^\alpha(2-p^2) dp.
$$
We find
\be
D^\alpha f_2(x)=\frac{2^{\,\alpha-2}}{\sqrt{\pi}}\,\biggl[\left(i^{\alpha}+(-i)^{\alpha}\right)  \Gamma \left(\frac{1+\alpha}{2}\right) \biggl(\, _1F_1\left(\frac{1+\alpha}{2},\frac{1}{2},-x^2\right)-(1+\alpha)\, _1F_1\left(\frac{3+\alpha}{2},\frac{1}{2},-x^2\right)\biggr)-
$$
$$
-2i\left((-i)^{\alpha}-i^{\alpha}\right)\,x  \Gamma \left(1+\frac{\alpha}{2}\right)\biggl(\, _1F_1\left(\frac{2+\alpha}{2},\frac{3}{2},-x^2\right)-(2+\alpha)\, _1F_1\left(\frac{4+\alpha}{2},\frac{3}{2},-x^2\right)\biggr)\biggr],
\label{32}\en
for all $\alpha\geq0$. Once again, it is possible to check that this result returns the ordinary derivatives for $\alpha=0,1,2,3,\ldots$. Figure \ref{fig3} shows the behaviour of $D^\alpha f_2(x)$ for values of $\alpha$ approaching zero: we see that, the closer $\alpha$ is to zero, the more $D^\alpha f_2(x)$ goes to $f_2(x)$, see Theorem \ref{th3}. Again, while $f_2(x)$ is even, $D^\alpha f_2(x)$ is not, if $\alpha$ is close, but not equal, to zero.

\begin{figure}[ht]
	\begin{center}
		\includegraphics[width=0.90\textwidth]{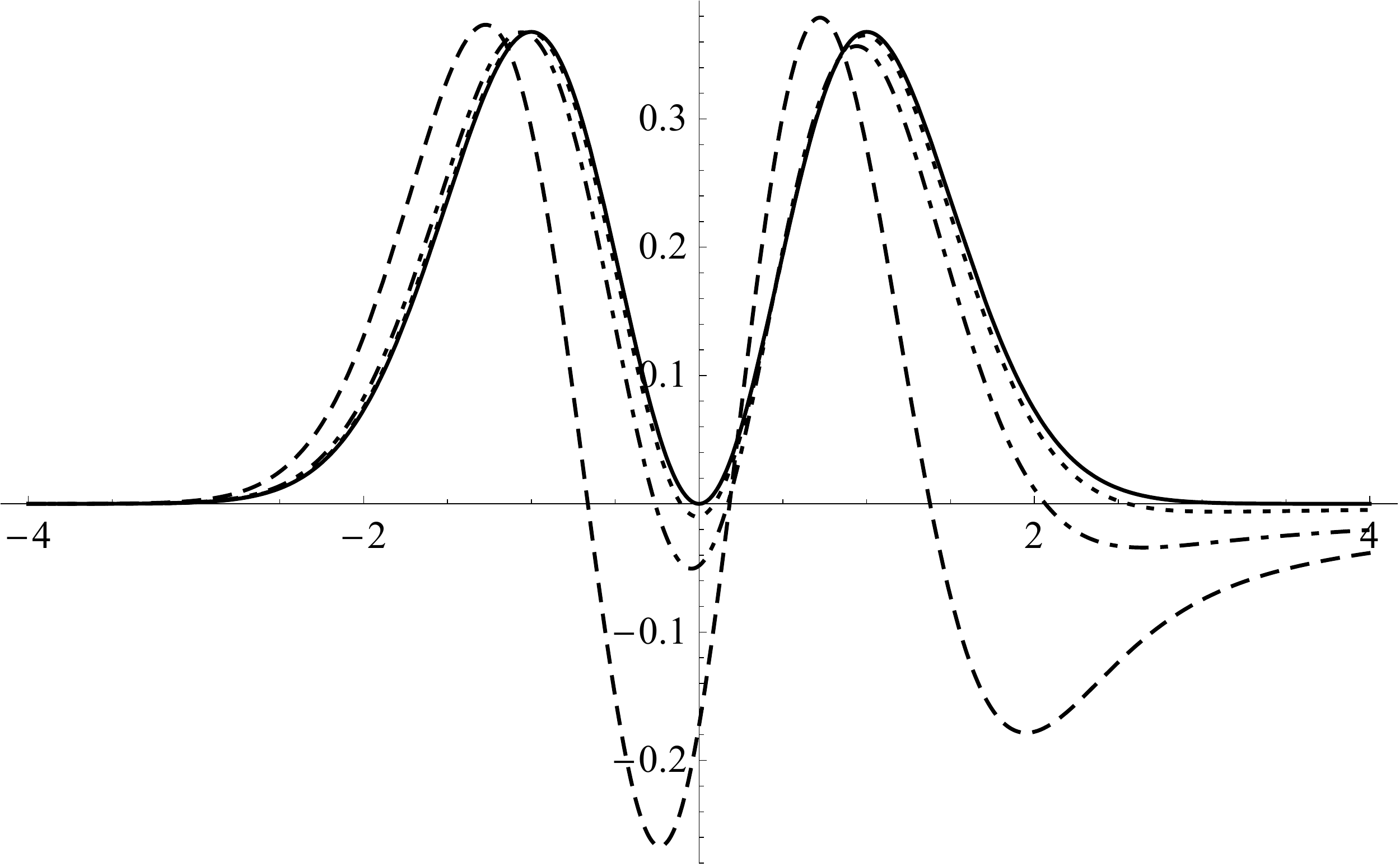}\hspace{8mm} %
	\end{center}
	\caption{{$D^\alpha f_2(x)$ for $\alpha=0$ (continuos line), $\alpha=\frac{1}{50}$ (dotted line), $\alpha=\frac{1}{10}$ (dotted-dashed line), and $\alpha=\frac{1}{2}$ (dashed line).}}
	\label{fig3}
\end{figure}

\subsection{Example 3: $f_3(x)=e^{x}$}

This example extends what is done previously, since $f_3(x)\notin\Sc(\mathbb{R})$, even if it is a $C^\infty$-function. Of course, $f_3(x)$ is not even an $\Lc^2(\mathbb{R})$-function. Nevertheless, using the same results in Section \ref{sect2} in a distributional sense, we can still compute the fractional derivative of $e^x$, just checking that $I_\alpha$ in (\ref{25bis}) is well defined. We first observe that the Fourier transform of $f_3(x)$ exists as a distribution: $\hat f_3(p)=\sqrt{2\pi}\,\delta(p+i)$. We refer to \cite{brew,lindell,smagin}, and to references therein, for possible definitions of the generalized Dirac delta distribution with complex argument, and some of their properties. Therefore, recalling that $g(x)\delta(x-x_0)=g(x_0)\delta(x-x_0)$ for all continuous function $g(x)$ (in a weak sense, even if $x_0\in\mathbb{C}$), $g_\alpha(p)=\sqrt{2\pi}\,(ip)^\alpha\,\delta(p+i)=\sqrt{2\pi}\,\delta(p+i)$. We see that $g_\alpha(p)$ does not depend on $\alpha$ anymore. Using formula (\ref{22}) we deduce that
$$
D^\alpha f_3(x)=\frac{1}{\sqrt{2\pi}}\int_{\mathbb R}e^{ipx}g_\alpha(p)\,dp=\int_{\mathbb R}e^{ipx}\delta(p+i)\,dp=e^x,
$$
which shows that not only all the standard integer derivatives of $e^x$ return $e^x$ itself, as they should, but also the fractional ones.

It is interesting to check that $f_3(x)$ satisfies the condition for the existence of fractal derivative as derived at the end of Section \ref{sect2}. In other words, $I_\alpha$ in (\ref{25bis}) is finite. In fact, in this case we have
$$
I_\alpha=\int_{\mathbb R}\overline{\hat f_3(p)}\,(ip)^\alpha \hat f(p)\,dp=\sqrt{2\pi}\,\int_{\mathbb R}\delta(p-i)\,(ip)^\alpha \hat f(p)\,dp=\sqrt{2\pi} (-1)^\alpha \hat f(i),
$$
which is clearly finite (in modulus), since $\hat f(p)\in\Sc(\mathbb{R})$.

It is easy to generalize this example by replacing $f_3(x)$ with ${\bf f}_3(x)=e^{kx}$. In this case ${\bf\hat f_3}(p)=\sqrt{2\pi}\,\delta(p+ik)$, and  $g_\alpha(p)=\sqrt{2\pi}\,(ip)^\alpha\,\delta(p+ik)=\sqrt{2\pi}\,k^\alpha\delta(p+ik)$, whose inverse Fourier transform returns 
$D^\alpha {\bf f_3}(x)=F^{-1}\left[g_\alpha(p)\right]=k^\alpha e^{kx}$, which is exactly the result assumed by Liouville in his treatment of the fractional derivatives, see also Table 5.1 in \cite{herr}.

\subsection{Example 4: $f_4(x)=x^n$}\label{example4}

This is, in fact, a class of functions, none of which belonging to $\Sc(\mathbb{R})$. Once again, therefore, our computations must be understood in the generalized sense described in Section \ref{sect2}. This can be done since  the integral  $I_\alpha$ in (\ref{25bis}) is finite for certain values of $\alpha$. 

In particular, let us see what happens for $n=1$. In this case, $f_4(x)=x$ and $\hat f_4(p)=i\sqrt{2\pi}\,\delta'(p)$, where $\delta'(p)$ is the weak derivative of $\delta(p)$. Then
$$
I_\alpha=\int_{\mathbb R}\overline{\hat f_4(p)}\,(ip)^\alpha \hat f(p)\,dp=i\sqrt{2\pi}\,\int_{\mathbb R}\frac{d}{dp}\delta(p)\,(ip)^\alpha \hat f(p)\,dp=-i\sqrt{2\pi}\,\int_{\mathbb R}\delta(p)\,\frac{d}{dp}\left((ip)^\alpha \hat f(p)\right)\,dp.
$$
It is easy to see then that, for instance, $I_\alpha=0$ if $\alpha>1$, while $I_1=\sqrt{2\pi}\hat f(0)$. In both cases, $I_\alpha$ is finite, so that, as discussed at the end of Section \ref{sect2}, the fractional derivative of $f_4(x)=x$ can be defined, at least if $\alpha=0$ (which is obvious), or when $\alpha\geq1$.

For generic $n$ the Fourier transform of $f_4(x)$ is  $\hat f_4(p)=i^n\sqrt{2\pi}\,\delta^{(n)}(p)$. Of course, here the derivative of the $\delta(p)$ is the weak one. Then $g_\alpha(p)=\sqrt{2\pi}\,(ip)^\alpha\,i^n\,\delta^{(n)}(p)$, and from (\ref{22}) we get
\be
D^\alpha f_4(x)=i^n\int_{\mathbb R}e^{ipx}(ip)^\alpha\,\delta^{(n)}(p)\,dp,
\label{33}\en
for those $\alpha\geq0$ for which the integral exists, and for all $n=0,1,2,3,\ldots$. Of course, we expect that the range of allowed values of $\alpha$ depends on our choice of $n$. This will appear clear in the rest of the section.  

Let us see what happens for some fixed values of $n$, beginning with $n=0$. In this case, $f_4(x)=1$, and we expect that all its derivatives (except for $\alpha=0$) are zero. In fact, from (\ref{33}), if $n=0$ we get
$
D^\alpha f_4(x)=\int_{\mathbb R}e^{ipx}(ip)^\alpha\,\delta(p)\,dp
$, which is equal to 1 if $\alpha=0$, and returns 0 for all other positive values of $\alpha$. Hence, if $n=0$, the fractional derivative of $f_4(x)=1$ exists for all $\alpha\geq0$.

Let us now fix $n=1$. In this case, using the definition of the weak derivative,
$$
D^\alpha f_4(x)=i\int_{\mathbb R}e^{ipx}(ip)^\alpha\,\delta^{(1)}(p)\,dp=-i\int_{\mathbb R}\frac{d}{dp}\left(e^{ipx}(ip)^\alpha\right)\,\delta(p)\,dp.
$$
Once again, we have to distinguish between the case $\alpha=0$ and $\alpha>0$. When $\alpha=0$ we easily get $$D^0 f_4(x)=-i\int_{\mathbb R}\frac{d}{dp}\left(e^{ipx}\right)\,\delta(p)\,dp=x\int_{\mathbb R}e^{ipx}\,\delta(p)\,dp=x,$$
as it should. If $\alpha>0$, after simple computations we get
$$
D^\alpha f_4(x)=-\alpha(i)^{\alpha+1}\int_{\mathbb R}e^{ipx}p^{\alpha-1}\delta(p)dp.
$$
Notice that, according to our preliminary analysis of $I_\alpha$, this integral exists if $\alpha\geq1$. In particular, $D^\alpha f_4(x)=1$ if $\alpha=1$, while $D^\alpha f_4(x)=0$ if $\alpha>1$.

We next take $n=2$. Hence $f_4(x)=x^2$. Because of the definition of weak derivative, we have
$$
D^\alpha f_4(x)=i^2\int_{\mathbb R}e^{ipx}(ip)^\alpha\,\delta^{(2)}(p)\,dp=-(-1)^2\int_{\mathbb R}\frac{d^2}{dp^2}\left(e^{ipx}(ip)^\alpha\right)\,\delta(p)\,dp.
$$
Simple computations show that $D^0 f_4(x)=0$, $D^1 f_4(x)=2x$, $D^2 f_4(x)=2$, and  $D^n f_4(x)=0$ for $n=3,4,5,\ldots$, as expected. 
For fractional $\alpha$, the situation is the following: since
$$
\frac{d^2}{dp^2}\left(e^{ipx}(ip)^\alpha\right)=e^{ipx}\left(-\alpha(\alpha-1)(ip)^{\alpha-2}-2\alpha x(ip)^{\alpha-1}-x^2(ip)^\alpha\right),
$$
it is clear that $D^\alpha f_4(x)$ is surely well defined if $\alpha\geq2$. In particular, $D^\alpha f_4(x)=2$ if $\alpha=2$, as already derived, while $D^\alpha f_4(x)=0$ for all $\alpha>2$, as it is reasonable.

\vspace{2mm}

It would be interesting to understand if the range of admissible values of $\alpha$ arising from our analysis are {\em the largest possible}, or if something can be done even  in the regions where, for instance, $I_\alpha$ is not defined. This aspect of our analysis is part of our future plans.

\section{The fractional momentum operator}\label{sect5}

Our applications of definition (\ref{22}) to quantum mechanics is very preliminary. We will propose here a fractional version of the momentum operator, finding its (generalized) eigenstates, and we will compute its commutator with the position operator, studying then  the uncertainty relation arising from this commutator. Also, we will briefly introduce fractional bosonic annihilation and creation operators.

The starting point is the following definition:
\be
P_\alpha=\left(-i D\right)^\alpha,
\label{51}\en
where $\alpha\geq0$. If $\alpha=1$ this is exactly the well known momentum operator. To compute its commutator with the position operator, $x$, we observe that
$$
F[xf(x)](p)=i\hat f'(p),
$$
for all $f(x)\in\Sc(\mathbb{R}$). Hence we have
$$
[D^\alpha,x]f(x)=D^\alpha(xf(x))-xD^\alpha(f(x))=\frac{1}{\sqrt{2\pi}}\int_{\mathbb R}(ip)^\alpha e^{ipx}\left(i\hat f'(p)\right)dp-\frac{x}{\sqrt{2\pi}}\int_{\mathbb R}(ip)^\alpha e^{ipx}\hat f(p)\,dp
$$
which after some easy computations produces the following interesting result:
\be
[D^\alpha,x]f(x)=\alpha (D^{\alpha-1}f)(x).
\label{52}\en
This equation is well defined for $\alpha=0$ (which is obvious, since $D^0=\1$), and for $\alpha\geq1$. We see that there is a range of values of $\alpha$, $\alpha\in]0,1[$, in which the meaning of (\ref{52})  is not evident. This is similar to what we have already seen in Section \ref{example4}. In particular, formula (\ref{52}) returns the standard result $[\frac{d^n}{dx^n},x]=n\frac{d^{n-1}}{dx^{n-1}}$, for all $n=0,1,2,3,\ldots$, with the understanding that the RHS of this formula is zero when $n=0$.

\vspace{2mm}

{\bf Remark:--} We refer for instance to \cite{adams} for some useful discussion on the case $\alpha\in]0,1[$.

\vspace{2mm}

Going back to the fractional momentum operator $P_\alpha$, from (\ref{52}) we deduce that
\be [x,P_\alpha]=i\alpha P_{\alpha-1},\label{53}\en
for $\alpha=0$ and $\alpha\geq1$. In particular, $[x,P_0]=0$, which is consistent with the fact that $P_0$ is just the identity operator $\1$. Also, if $\alpha=1$ we get $[x,P_1]=i\1$, which is the well known result for the commutator between position and momentum operators in quantum mechanics. It is possibly useful to stress that formula (\ref{53}) must be understood, as in (\ref{52}), acting on a suitable function $f(x)$, which as before we assume belonging to $\Sc(\mathbb{R})$.

If we consider the following linear combinations of $x$ and $P_\alpha$, $A_\alpha=\frac{x+iP_\alpha}{\sqrt{2}}$ and $B_\alpha=\frac{x-iP_\alpha}{\sqrt{2}}$, we deduce that
\be
[A_\alpha,B_\alpha]=\alpha P_{\alpha-1},
\label{54} \en
which is again well defined for $\alpha=0$ and for $\alpha\geq1$. In particular, if $\alpha=1$ we know that $P_1=P_1^\dagger$, and therefore $B_\alpha=A_\alpha^\dagger$. Hence (\ref{54}) returns the standard commutation relation for bosonic operators. Otherwise we get something different, which could be worth of a deeper analysis. Once again, (\ref{54}) must be understood in the sense of unbounded operators, in principle: both sides of the equation must act on functions of some suitable domain, as those in $\Sc(\mathbb{R})$, for instance.

Let us go back to (\ref{53}). Here $x$ is self-adjoint and $P_\alpha$ satisfies (\ref{57}) below, which suggests that $P_\alpha$ is self-adjoint as well (which is surely true at least if $\alpha=0,1,2,3,\ldots$). We want to see then what the Heisenberg uncertainty principle implies for these operators. For that we recall that, if $A$ and $B$ are two self-adjoint operators such that $[A,B]=iC$, then $\Delta A\,\Delta B\geq \frac{|\left<C\right>|}{2}$. Here, for a generic self-adjoint operator $X$, we have put
$$
\left<X\right>=\left<\varphi,X\varphi\right>, \qquad (\Delta X)=\left<\varphi,(X-\left<X\right>)^2\varphi\right>,
$$
where $\varphi(x)$ is a fixed normalized function in $\Lc^2(\mathbb{R})$ for which these quantities are well defined. Then, at least for those values of $\alpha\geq1$ for which $P_\alpha$ is self-adjoint, we should have
\be
\Delta x\,\Delta P_\alpha\geq \alpha \frac{|\left<P_{\alpha-1}\right>|}{2}.
\label{58}
\en
In particular, if $\alpha=1$, this inequality becomes the well known $\Delta x\,\Delta P_1\geq  \frac{1}{2}$, since $P_0=\1$. Let us now see what the right-hand side of (\ref{58}) becomes for a particular choice of $\varphi(x)$. To fix the ideas, we take $\varphi(x)=\left(\frac{2}{\pi}\right)^{1/4}e^{-x^2}$, which is normalized in $\Lc^2(\mathbb{R})$ and is proportional to $f_1(x)$ introduced before. Using (\ref{31}), and the fact that
$$
\int_{\mathbb R}e^{-x^2}{}_1F_1\left(\frac{\alpha}{2},\frac{1}{2},-x^2\right)=\sqrt{\frac{\pi}{2^\alpha}}, \quad\mbox{and}\quad \int_{\mathbb R}e^{-x^2}\,x\,{}_1F_1\left(\frac{1+\alpha}{2},\frac{3}{2},-x^2\right)=0,
$$
we deduce that
\be
\Delta x\,\Delta P_\alpha\geq \alpha\frac{2^{(\alpha-3)/2}}{\sqrt{\pi}}\Gamma\left(\frac{\alpha}{2}\right)\left|\cos\left(\frac{(\alpha-1)\pi}{2}\right)\right|,
\label{59}
\en
which returns the standard result if $\alpha=1$. Figure \ref{fig4} shows the behaviour of the right-hand side (RHS) of this inequality for the interval  $\alpha\in[0,6]$: we see that this quantity is zero for $\alpha=2k$, $k=0,1,2,3,\ldots$, but we also see that this quantity can be very large for some values of $\alpha$. For these values becomes impossible to know simultaneously $x$ and $P_\alpha$, without getting a large uncertainty. It is worth to stress that the relevant part of the plot in Figure \ref{fig4} is that for $\alpha\geq1$, since for $\alpha<1$ formula (\ref{53}) makes no sense.

\begin{figure}[ht]
	\begin{center}
		\includegraphics[width=0.80\textwidth]{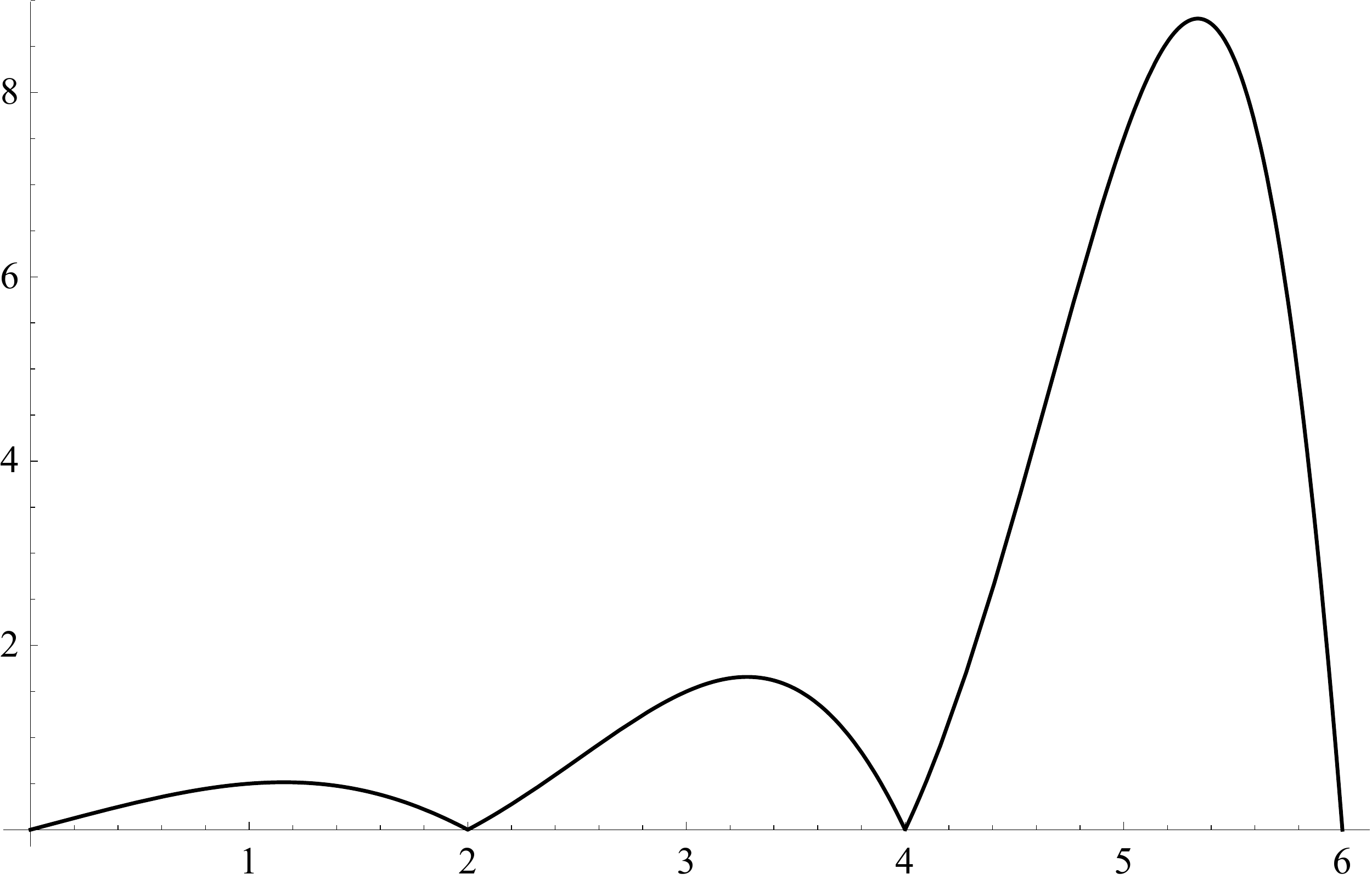}\hspace{8mm} %
	\end{center}
	\caption{The RHS of equation (\ref{59}) as a function of $\alpha$.}
	\label{fig4}
\end{figure}

The next question we want to address is the nature of the eigenstates of $P_\alpha$, if any. It is well known, for instance, see \cite{bag1,bag2} and references therein, that $P_1$ admits plane waves as eigenfunctions, and these are not square integrable. We will show here that the same conclusion can be deduced also for $\alpha\neq1$. 

Let us call $f_\alpha(x)$ the eigenstate of $P_\alpha$ with eigenvalue $E_\alpha$: $(P_\alpha f_\alpha)(x)=E_\alpha f_\alpha(x)$. Using (\ref{22}) this can be rewritten, after some algebra, as
$$
F^{-1}\left[p^\alpha \hat f_\alpha(p)\right]=E_\alpha f_\alpha(x)=E_\alpha F^{-1}\left[\hat f_\alpha(p)\right].
$$
It is natural to look for a solution $f_\alpha(x)$ satisfying the equality $p^\alpha \hat f_\alpha(p)=E_\alpha \hat f_\alpha(p)$. This equation admits the  weak solution $\hat f_\alpha(p)=N_\alpha\delta(p^\alpha-E_\alpha)$, where $N_\alpha$ is some proportionality constant. Then, 
\be
f_\alpha(x)=\frac{1}{\sqrt{2\pi}}\int_{\mathbb{R}} e^{ipx}\hat f_\alpha(p)dp=\frac{N_\alpha}{\sqrt{2\pi}}\int_{\mathbb{R}} e^{ipx}\delta(p^\alpha-E_\alpha)dp.
\label{55}\en
If $\alpha=1$, this formula easily produces $f_1(x)=\frac{N_1}{\sqrt{2\pi}}e^{iE_1x}$, which is a plane wave. Of course, $E_1$ can be any real number. Hence, the standard result for the {\em integer} momentum operator is recovered. Incidentally we observe that, if $\alpha=0$, $P_0=\1$, and therefore every non zero function (or distribution) is an eigenstate of $P_0$, with eigenvalue one. Let us now see what happens if $\alpha=2$. In this case $P_2$ is a positive, self-adjoint, operator. Hence its eigenvalue $E_2$ must be non negative, $E_2\geq0$. Well known formulas for the delta-function imply that
$$
\hat f_2(p)=N_2\delta(p^2-E_2)=\frac{N_2}{2\sqrt{E_2}}\left(\delta(p-\sqrt{E_2})+\delta(p+\sqrt{E_2})\right), 
$$
so that $f_2(x)=F^{-1}[\hat f_2(p)]=\frac{N_2}{\sqrt{2\pi}}\frac{\cos(\sqrt{E_2}\,x)}{\sqrt{E_2}}$. It is clear that $f_2(x)\notin\Lc^2(\mathbb{R})$, and $E_2$ can be any non negative real number. It is not difficult to imagine how the situation goes if $\alpha=3,4,\ldots$.

Let us now go to fractional values of $\alpha$. In particular, it is easier to consider $\alpha$ of the following form: $\alpha=\frac{1}{2n+1}$, $n=0,1,2,3,\ldots$. In the rest of this section we will restrict to these values. Let $h(p)\in\Sc(\mathbb{R})$ be a test function. Then, with the change of variable $p\rightarrow q=p^\frac{1}{2n+1}$, we see that, for each $\gamma\in\mathbb{R}$,
$$
\int_{\mathbb R}\delta\left(p^\frac{1}{2n+1}-\gamma\right)h(p)dp=(2n+1)\gamma^{2n}\int_{\mathbb R}\delta\left(q-\gamma^{2n+1}\right)h(q)\,dq,
$$
and therefore $\delta\left(p^\frac{1}{2n+1}-\gamma\right)=(2n+1)\gamma^{2n}\delta\left(q-\gamma^{2n+1}\right)$. This means that
$$
\hat f_{\frac{1}{2n+1}}(p)=N_{\frac{1}{2n+1}}(2n+1)E_{\frac{1}{2n+1}}^{2n}\delta\left(p-E_{\frac{1}{2n+1}}^{2n+1}\right),
$$
so that
\be
f_{\frac{1}{2n+1}}(x)=F^{-1}\left[\hat f_{\frac{1}{2n+1}}(p)\right]=\frac{(2n+1)N_{\frac{1}{2n+1}}(2n+1)E_{\frac{1}{2n+1}}^{2n}}{\sqrt{2\pi}}\,\exp\left\{iE_{\frac{1}{2n+1}}^{2n}x\right\},
\label{56}\en
which is again a plane wave, for all possible $n$. Incidentally we observe that there is no constraint on $E_{\frac{1}{2n+1}}$, which could also be, in principle, complex. However, we do not expect this can happen. This is suggested (but not proved!) by equation (\ref{23b}), which implies that, for all $\alpha\geq0$,
\be
\left<P_\alpha g,f\right>=\left<g,P_\alpha f\right>,
\label{57}\en
for all $f(x),g(x)\in\Sc(\mathbb{R})$. Hence, at least thinking for a moment that the eigenstate $f_\alpha(x)$ of $P_\alpha$ belongs to $\Sc(\mathbb{R})$ (which, as we have seen, is not the case!),  we would get $\left<P_\alpha f_\alpha,f_\alpha\right>=\left<E_\alpha f_\alpha,f_\alpha\right>=\overline{E_\alpha}\,\|f_\alpha\|^2$, as well as  $\left< f_\alpha,P_\alpha f_\alpha\right>=\left< f_\alpha,E_\alpha f_\alpha\right>={E_\alpha}\,\|f_\alpha\|^2$. Hence extending (\ref{57}) to $f_\alpha$, $\left<P_\alpha f\alpha,f_\alpha\right>=\left<f_\alpha,P_\alpha f_\alpha\right>$ only if $E_\alpha\in\mathbb{R}$.

\section{Conclusions}\label{sect6}

We have investigated some properties of the fractional derivative defined via Fourier transform and distribution theory. We have given some examples, and considered a preliminary application to quantum mechanics.

Several interesting aspects  should be considered in a future analysis: is the range of $\alpha$ for which $D^\alpha$ can be defined strictly connected to the existence of $I_\alpha$? Is the role of $\Sc(\mathbb{R})$ so important? Can  definition (\ref{22}) be extended further? Also, from a more applicative side, an interesting question arises: is it possible, in analogy with what is done for ordinary bosonic operators, satisfying the canonical commutation relations, or for pseudo-bosonic operators, \cite{baginbagbook}, to set up an algebraic approach for the operators $A_\alpha$ and $B_\alpha$ in (\ref{54}) to produce an orthonormal basis of the Hilbert space where the operators act, or (maybe) two complete families of biorthogonal functions? This could be relevant in connection with particular deformed canonical commutation relations in PT-quantum mechanics. We hope to be able to consider this aspect in a future analysis.

%\newpage
\section*{Acknowledgements}

This work was partially supported by the University of Palermo and by the Gruppo Nazionale di Fisica Matematica of Indam.

\end{document}